\documentclass{sig-alternate}
\newcommand{\E}{\text{E}}
\usepackage{amsmath, amsthm-alt, graphicx, cite, float, hyperref}
\usepackage[lined,boxed]{algorithm2e}
\usepackage[labelfont=bf,font=small]{caption}

\newtheorem{thm}{Theorem}

\newcommand{\vs}{\vspace{0.1cm}}

\makeatletter
\def\@copyrightspace{\relax}
\makeatother

\begin{document}\sloppy
\title{Recommendation Subgraphs for Web Discovery\thanks{This work was supported in part by an internship at BloomReach Inc.}}

\numberofauthors{3}
\author{
\alignauthor
Arda Antikacioglu~\thanks{Supported in part by NSF CCF-1347308} \\
       \affaddr{Department of Mathematics}\\
       \affaddr{Carnegie Mellon University}\\
       \email{aantikac@andrew.cmu.edu}
\alignauthor
R. Ravi~\thanks{Supported in part by NSF CCF-1347308}  \\
       \affaddr{Tepper School of Business}\\
       \affaddr{Carnegie Mellon University}\\
       \email{ravi@cmu.edu}
\alignauthor
Srinath Sridhar \\
       \affaddr{BloomReach Inc.}\\
       \email{srinath@bloomreach.com}
}

\maketitle
\abstract

Recommendations are central to the utility of many websites including
YouTube, Quora as well as popular e-commerce stores. Such sites
typically contain a set of recommendations on every product page that 
enables visitors to easily navigate the website. Choosing an
appropriate set of recommendations at each page is one of the key 
features of backend engines that have been deployed at several e-commerce
sites.\vs

Specifically at BloomReach, an engine consisting of several
independent components analyzes and optimizes its clients
websites. This paper focuses on the structure optimizer component
which improves the website navigation experience that enables the
discovery of previously undiscovered content.\vs

We begin the paper by formalizing the concept of recommendations used
for discovery. We formulate this as a natural graph optimization
problem which in its simplest case, reduces to a bipartite matching
problem. In practice, solving these matching problems requires
superlinear time and is not scalable. Also, implementing simple
algorithms is critical in practice because they are significantly
easier to maintain in a production software package. This motivated us
to analyze three methods for solving the problem in increasing order
of sophistication: a local random sampling algorithm, a greedy algorithm
and a more involved partitioning based algorithm. \vs

We first theoretically analyze the performance of these three methods
on random graph models characterizing when each method will yield a
solution of sufficient quality and the parameter ranges when more 
sophistication is needed. We complement this by providing an empirical
analysis of these algorithms on simulated and real-world production
data. Our results confirm that it is not always necessary to implement
complicated algorithms in the real-world. Indeed, our results
demonstrate that very good practical results can be obtained by 
using simple heuristics that are backed by the confidence of concrete
theoretical guarantees. \vs
\setcounter{page}{1}
\section{Introduction}

\subsection{Web Relevance Engines}
The digital discovery divide~\cite{WebRelevanceEngine} refers to the problem of companies not being able to present users with what they seek in the short time they spend looking for this information. The problem is prevalent not only in e-commerce websites but also in social networks and micro-blogging sites where surfacing relevant content quickly is important for user engagement. \vs

BloomReach is a big-data marketing company that uses the client's content as well as web-wide data to optimize both customer acquisition and satisfaction for e-retailers.
BloomReach's clients include popular retailers like Nieman Marcus, Crate \& Barrel, Williams-Sonoma and Staples besides many others. In this paper, we describe the structure optimizer component of BloomReach's Web Relevance Engine. This component works on top of the recommendation engine so as to carefully add a set of links across pages that ensures that users can efficiently navigate the entire website.

\subsection{Structure Optimization of Websites}

One of the great benefits of the web as a useful source of hyperlinked
information comes from the careful choices made in crafting the
recommendations that link a page to closely related pages. Though this
advantage was identified well before the WWW was in place by
Bush~\cite{Bush45aswe}, it continues to persist even today.
Recent estimates~\cite{big-data-book13} attribute up to a third of the sales
on Amazon and three-quarters of new orders on Netflix to users that are
influenced by the carefully chosen recommendations provided to them. \vs

Even though recommendations exist across the entire web, we provide some simple concrete examples. First, YouTube has a section that displays all the related videos for every main video being viewed.
Quora has a section for questions related to the main question that is displayed. These recommendations are critical in determining how the traffic across all of YouTube or Quora is going to flow.
An important concern of website owners is whether a significant fraction of the site is not recommended at all (or `hardly' recommended) from other more popular pages. Continuing
with the above example, if a large fraction of the YouTube videos were not recommended from any (or few) other videos, then millions of great videos will lie undiscovered. One way to address this
problem is to try to ensure that every page will obtain at least a baseline number of visits so that great content does not remain undiscovered, and thus bridge the discovery divide mentioned above. \vs

We use the criterion of discoverability as the objective for the choice of the links to recommend. Consequently, we get a new formulation of the recommendation selection problem that is structural. In particular, we think of commonly visited pages in a site as the already discovered pages, from which there are a large number of possible recommendations available to related but less visited peripheral pages. The problem of choosing a limited number of pages to recommend at each discovered page can be cast with the objective of maximizing the number of peripheral non-visited pages that are linked. We formulate this as a recommendation subgraph problem, and study practical algorithms for solving these problems in real-life data. \vs

\subsection{Recommendation Systems as a Subgraph Selection Problem}

Formally, we divide all pages in a site into two groups: the discovered pages and the undiscovered ones.
Furthermore, we assume that recommendation systems~\cite{Schafer1999, Adomavicius2005,
  Resnick1997} find a large set of related candidate undiscovered page recommendations
for each discovered page using relevance. In this work, we assume $d$
such related candidates are available per page creating a candidate recommendation bipartite graph 
(with degree $d$ at each discovered page node).
Our goal is to analyze how to prune this set to $c < d$ recommendations such that
globally we ensure that the resulting recommendation subgraph can be navigated efficiently by the user to enable better discovery. \vs

\subsection{Our Contributions}
While optimal solutions to some versions of the recommendation subgraph problem can be obtained by using a maximum matching algorithm, such algorithms are too costly to run on real-life instances. We introduce three simple alternate methods that can be implemented in linear or near-linear time and examine their properties. 
In particular, we delineate when
each method will work effectively on popular random graph models, and when a practitioner will need to employ a more sophisticated algorithm. 
We then evaluate how these simple methods perform on simulated data, both in terms of solution quality and  running time.
Finally, we show the deployment of these methods on BloomReach's real-world client link graph and measure their
actual performance in terms of running-times, memory usage and accuracy. \vs

To summarize, our contributions are as follows.
\begin{enumerate}
\item The development of a new structural model for recommendation systems as a subgraph selection problem for maximizing discoverability, 
\item The proposal of three methods with increasing sophistication to solve the problem at scale along with associated theoretical performance guarantee analysis for each method, and
\item An empirical validation of our conclusions with simulated and real-life data.
\end{enumerate}

\section{Related Work}

Recommendation systems have been studied extensively in literature, especially since the advent of the web. Most recommendation systems can be broadly separated into two different groups: collaborative filtering systems and content-based recommender systems \cite{almazro2010survey}. Much attention has been focused on the former approach, where either users are clustered by considering the items they have consumed or items are clustered by considering the users that have bought them. Both item-to-item and user-to-user recommendation systems based on collaborative filtering have been adopted by many industry giants such as Twitter \cite{twitter-collab-filtering}, Amazon \cite{amazon-collab-filtering} and Google \cite{google-collab-filtering}.  \vs

Content based systems instead look at each item and its intrinsic properties. For example, Pandora has categorical information such as Artist, Genre, Year, Singer, Tempo etc. on each song it indexes. Similarly, Netflix has a lot of categorical data on movies and TV such as Cast, Director, Producers, Release Date, Budget, etc. This categorical data can then be used to recommend new songs that are similar to the songs that a user has liked before. Depending on user feedback, a recommender system can learn which of the categories are more or less important to a user and adjust its recommendations. \vs

A drawback of the first type of system is that is that they require multiple visits by many users so that a taste profile for each user, or a user profile for each item can be built. 
Similarly, content-based systems also require significant user participation to train the underlying system. These conditions are possible to meet for large commerce or entertainment hubs such as the companies mentioned above, but not very likely for most online retailers that specialize in a just a few areas, but have a long-tail~\cite{Anderson2006} of product offerings. \vs

Because of this constraint, in this paper we focus on a recommender system that typically uses many different algorithms that extract categorical data from item descriptions and uses this data to establish weak links between items (candidate recommendations). In the absence of data that would enable us to choose among these many links, we consider every potential recommendation to be of equal value and focus on the objective of discovery, which has not been studied before. In this way, our work differs from all the previous work on recommendation systems that emphasize on finding recommendations of high relevance and quality rather than on structural navigability of the realized link structure. However, while it's not included in this paper for brevity, some of our approaches can be extended to the more general case where different recommendations have different weights. \vs
\section{Our Model}

We model the structure optimization of recommendations by using a bipartite
digraph, where one partition $L$ represents the set of discovered
(i.e., often visited) items for which we are required to suggest recommendations and the other partition $R$
representing the set of undiscovered (not visited) items that can be potentially recommended. If
needed, the same item can be represented in both $L$ and $R$.
\vs

\subsection{The Recommendation Subgraph Problem}
We introduce and study this as the {\bf the $(c, a)$-recommendation subgraph problem} in this paper:
{\em
 The input to the problem is the graph where each
$L$-vertex has $d$ recommendations. Given the space restrictions to
display recommendations, the output is a subgraph where each vertex in
$L$ has $c < d$ recommendations. The goal is to maximize the number of
vertices in $R$ that have in-degree at least a target integer $a$.
}

\vs

Note that if $a=c=1$ this is simply the maximum bipartite
matching problem~\cite{LovaszPlummer1986}. If $a=1$ and $c > 1$, we
obtain a $b$-matching problem, that can be converted to a bipartite
matching problem~\cite{Gabow1983}.\vs

We now describe typical web graph characteristics by discussing the
sizes of $L$, $R$, $c$ and $a$ in practice. As noted before, in most
websites, a small number of `head' pages contribute to a significant
amount of the traffic while a long tail of the remaining pages
contribute to the rest~\cite{HubermanAdamic1999,
  DuDemmerBrewer2006, KumarNorrisSun2009}. As demonstrated by a prior
measurement. This is supported by our
own experience with the 80/20 rule, i.e. 80\% of a site's traffic is
captured by 20\% of the pages. Therefore, the ratio $k=|L|/|R|$ is
typically between $1/3$ to $1/5$, but may be even lower. \vs

By observing recommendations of Quora, Amazon, YouTube and our own
work at BloomReach, typical values for $c$ range from 3 to 20
recommendations per page. Values of $a$ are harder to nail down but it
typically ranges from $1$ to $5$. 

\subsection{Practical Requirements}

There are two key requirements in making graph algorithms
practical. The first is that the method used must be very simple to
implement, debug, deploy and most importantly maintain long-term. The second is that the method must scale
gracefully with larger sizes. \vs

Graph matching algorithms require linear memory and super-linear run-time
which does not scale well. For example, an e-commerce website of a
client of BloomReach with 1M product pages and 100 recommendation
candidates per product would require easily over 160GB in main memory to store the graph
and run matching algorithms; this can be reduced by using graph
compression techniques but that adds more technical difficulties in
development and maintenance. Algorithms that are time intensive
can sometimes be sped-up by using distributed computing techniques such as
map-reduce~\cite{DeanGhemawat2004}. However, efficient map-reduce
algorithms for graph problems are notoriously difficult. \vs

\subsection{Simple Solutions}

To circumvent the time and space complexity of implementing optimal
graph algorithms for the recommendation subgraph problem, we propose
the study of three simple solutions strategies that not only can be
shown to scale well in practice but also have good theoretical
properties that we demonstrate using approximation ratios.

\begin{itemize}

\item {\bf Sampling:} The first solution is a simple random sampling
  solution that selects a random subset of $c$ links out of the
  available $d$ from every page. Note that this solution requires no
  memory overhead to store these results a-priori and the
  recommendations can be generated using a random number generator on
  the fly. While this might seem trivial at first, for sufficient (and
  often real-world) values of $c$ and $a$ we show that this can be
  optimal. Furthermore, while we omit this result for brevity, our 
  approach can be extended to the case where the recommendation edges
  have graphs.

\item {\bf Greedy:} The second solution we propose is a greedy
  algorithm that chooses the recommendation links so as to maximize
  the number of nodes in $R$ that can accumulate $a$ in-links. In
  particular, we keep track of the number of in-links required for
  each node in $R$ to reach the target of $a$ and choose the links
  from each node in $L$ giving preference to adding links to nodes in
  $R$ that are closer to the target in-degree $a$.

\item {\bf Partition:} The third solution is inspired by a
  theoretically rigorous method to find optimal subgraphs in
  sufficiently dense graphs: it partitions the edges into $a$ subsets
  by random sub-sampling, such that there is a good chance of finding
  a perfect matching from $L$ to $R$ in each of the subsets. The union
  of the matchings so found will thus result in most nodes in $R$
  achieving the target degree $a$. We require the number of edges in
  the underlying graph to be significantly large for this method to
  work very well; moreover, we need to run a (near-)perfect matching
  algorithm in each of the subsets which is also a computationally
  expensive subroutine. Hence, even though this method works very well
  in dense graphs, it does not scale very well in terms of running
  time and space.
\end{itemize}

As a summary, the table below shows the time and space complexity
of our different algorithms.

\begin{figure}[H]
\centering
\begin{tabular}{l|l|l|l|l}
\cline{2-4}
                                    & Sampling                                        & Greedy   & Partition          &  \\ \cline{1-4}
\multicolumn{1}{|l|}{Time}          & $O(|E|)$                                        & $O(|E|)$ & $O(|E|\sqrt{|V|})$ &  \\ \cline{1-4}
\multicolumn{1}{|l|}{Working Space} & $O(1)$                                          & $O(V)$   & $O(|E|)$           &  \\ \cline{1-4}
\end{tabular}
\caption{Complexities of the different algorithms (assuming constant $a$ and $c$)}
\end{figure}

In the next section, we elaborate on these methods, their running
times, implementation details, and theoretical performance
guarantees. In the section after that, we present our comprehensive
empirical evaluations of all three methods, first the results on
simulated data and then the results on real data from some clients of
BloomReach. 
\section{Algorithms for Recommendation Subgraphs}

\subsection{The Sampling Algorithm}
\label{fixed-degree}

We present the sampling algorithm for the $(c,a)$-recommendation subgraph formally below.

\begin{algorithm}
  \SetAlgoLined
  \KwData{A bipartite graph $G=(L,R,E)$}
  \KwResult{A $(c,a)$-recommendation subgraph $H$}
  \For{$u$ in $L$}{
    $S \leftarrow$ a random sample of $c$ vertices without replacement in $N(u)$\;
    \For{$v$ in S} {
      $H \leftarrow H \cup \{(u,v)\}$\;
   	}
  }
  \Return $H$\;
  \caption{The sampling algorithm}
  \label{sampling-algo}
\end{algorithm} \vs

Given a bipartite graph $G$, the algorithm has runtime complexity of
$O(|E|)$ since every edge is considered at most once. The space
complexity can be taken to be $O(1)$, since the adjacency representation
of $G$ can be assumed to be pre-sorted by the endpoint of each edge in $L$.\vs

We next introduce a simple random graph model for the supergraph from
which we are allowed to choose recommendations and present a bound on
its expected performance when the underlying supergraph $G=(L,R,E)$ is
chosen probabilistically according to this model.
\vs

{\bf Fixed Degree Model:} In this model for generating the candidate
recommendation graph, each vertex $u\in L$ uniformly and independently
samples $d$ neighbors from $R$ with replacement. While this allows each
vertex in $L$ to have the same vertex as a neighbor multiple times, in
reality $r \gg d$ is so edge repetition is very unlikely.
This model is similar to, but is distinct from the more
commonly known Erd\"{o}s-Renyi model of random graphs~\cite{Janson2011}.
In particular, while the degree of each vertex in $L$ is fixed under
our model, concentration bounds can show that the degrees of the
vertices in $L$ would have similarly been concentrated around $d$ for $p=d/r$
in the Erd\"{o}s-Renyi model. We prove 
the following theorem about the performance of the Sampling Algorithm.
We denote the ratio of the size of $L$ and $R$ by $k$, i.e., we define
$k = \frac{l}{r}$.

\begin{thm}\label{original_result}
Let $S$ be the
random variable denoting the number of vertices $v \in R$ such that
$\deg_{H}(v)\geq a$ in the fixed-degree model. Then
\[ \emph{\E}[S] \geq r\left(1-e^{-ck+\frac{a-1}{r}}\frac{(ck)^a-1}{ck-1}\right)  \]
\end{thm}

\begin{proof}
We will analyze the sampling algorithm as if it picks the neighbors of
each $u\in L$ with replacement, the same way the fixed-degree model
generates $G$. This variant would obviously waste some edges, and perform
worse than the variant which samples neighbors without replacement. This
means that any performance guarantee we prove for this variant holds
for our original statement of the algorithm as well. \vs

To prove the claim let $X_{v}$ be the random variable that represents
the degree of the vertex $v\in R$ in our chosen subgraph $H$. Because our
algorithm uniformly subsamples a uniformly random selection of edges,
we can assume that $H$ was generated the same way as $G$ but sampled $c$
instead of $d$ edges for each vertex $u\in L$. Since there are $cl$
edges in $H$ that can be incident on $v$, and each of these edges has a 
$1/r$ probability of being incident on a given vertex in $L$, we can now
calculate that

\begin{align*}
      \Pr[X_v = i]
&=    \binom{cl}{i} (1-\frac{1}{r})^{cl-i} \left(\frac{1}{r}\right)^i \\
&\leq (cl)^i (1-\frac{1}{r})^{cl-i} \left(\frac{1}{r}\right)^i
\end{align*}

Using a union bound, we can combine these inequalities to upper bound
the probability that $\deg_H(v)<a$.

\begin{align*}
      \Pr[X_v < a]
&=    \sum_{i=0}^{a-1} \binom{cl}{i} \left(1-\frac{1}{r}\right)^{cl-i}\left(\frac{1}{r}\right)^i \\
&\leq \sum_{i=0}^{a-1} \left(\frac{cl}{r}\right)^i\left(1-\frac{1}{r}\right)^{cl-i} \\
&\leq    \left(1-\frac{1}{r}\right)^{cl-(a-1)} \sum_{i=0}^{a-1} (ck)^i \\
&\leq \left(1-\frac{1}{r}\right)^{cl-(a-1)}\frac{(ck)^a-1}{ck-1} \\
&\leq e^{-ck+\frac{a-1}{r}} \frac{(ck)^a-1}{ck-1}
\end{align*}

Letting $Y_v = \left[X_v \geq a\right]$, we now see that

\[ \E[S] = \E\left[\sum_{v\in R} Y_v\right] \geq r\left(1-e^{-ck+\frac{a-1}{r}} \frac{(ck)^a-1}{ck-1}\right) \]
\end{proof} \vspace{-.2cm}

We can combine this lower bound with a trivial upper bound to obtain an
approximation ratio that holds in expectation.  

\begin{thm}
The above sampling algorithm gives a $\left(1-\frac1e\right)$-factor approximation to the $(c,1)$-graph recommendation problem in expectation.
\end{thm}
\begin{proof}
The size of the optimal solution is bounded above by both the number
of edges in the graph and the number of vertices in $R$. The former of
these is $cl=ckr$ and the latter is $r$, which shows that the optimal solution size
$OPT \leq
r\max(ck,1)$. Therefore, by simple case analysis the approximation ratio
in expectation is at least $({1-\exp(-ck)})/\min(ck,1) \geq 1-\frac{1}{e} $
\end{proof}

For the $(c, 1)$-recommendation subgraph problem the approximation obtained by this sampling approach can be much better for certain values of $ck$. In particular,
if $ck>1$, then the approximation ratio is $1-\exp(-ck)$, which
approaches 1 as $ck\to\infty$. When $ck=3$, then the
solution will be at least 95\% as good as the optimal solution even
with our trivial bounds. Similarly, when $ck<1$, the approximation
ratio is $(1-\exp(-ck))/ck$ which also approaches 1 as $ck\to 0$. In
particular, if $ck=0.1$ then the solution will be at 95\% as good as
the optimal solution. The case when $ck=1$ represents the
worst case outcome for this model where we only guarantee 63\%
optimality. Figure~\ref{fig:simple_approx} shows the approximation ratio as a
function of $ck$ for the $(c,1)$-recommendation subgraph problem in the fixed degree model.\vs

\begin{figure}[H]
  \centering
  \includegraphics[width=0.34\textwidth]{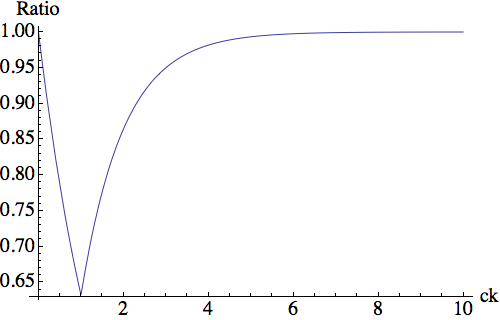}
  \caption{Approx ratio as a function of $ck$ }\label{fig:simple_approx}
\end{figure}

For the general $(c, a)$-recommendation subgraph problem, if $ck>a$,
then the problem is easy on average. This is in comparison to the
trivial estimate of $cl$. For a fixed $a$, a random solution gets
better as $ck$ increases because the decrease in $e^{-ck}$ more than
compensates for the polynomial in $ck$ next to it. However, in the
more realistic case, the undiscovered pages in $R$ too numerous to be all covered even if we used the full set of budgeted links allowed out of $L$, i.e. $cl < ra$ or rearranging, $ck<a$; in this case, we need to use the trivial estimate of
$ckr/a$, and the analysis for $a=1$ does not extend here. For practical purposes, the table
in Figure~\ref{a-values} shows how large $c$ needs to be (in terms of $k$) for the
solution to be 95\% optimal for different values of $a$, again in the fixed degree model.\vs

\begin{figure}[H]
  %\vspace{.2cm}
  \centering
  \begin{tabular}{ |c|c|c|c|c|c| }
    \hline
    $a$ & 1 & 2 & 3 & 4 & 5 \\ \hline
    $c$ & $3.00k^{-1}$ & $4.74k^{-1}$ & $7.05k^{-1}$ & $10.01k^{-1}$ & $13.48k^{-1}$ \\
    \hline
  \end{tabular}
  \caption{The required $ck$ to obtain 95\% optimality for $(c, a)$-recommendation subgraph}
  \label{a-values}
\end{figure}

We close out this section by showing that the main result that holds
in expectation also hold with high probability for $a=1$, using the
following variant of Chernoff bounds.

\begin{thm}\label{negative_corr_chernoff}~\cite{AugerDoerr2011}
Let $X_1,\ldots, X_n$ be non-positively correlated variables. If $X=\sum_{i=1}^n X_i$, then for any $\delta\geq 0$
\[ \Pr[X \geq (1+\delta)\emph{\E}[X] ] \leq \left(\frac{e^\delta}{(1+\delta)^{1+\delta}}\right)^{\E[X]} \]
\end{thm}

%Using this we can convert our expectation result to one that holds
%with high probability.

\begin{thm}
Let $S$ be the random variable denoting the number of vertices $v \in R$ such that $\deg_{H}(v)\geq 1$. Then
$ S \leq r(1-2\exp(-ck))$ with probability at most $(e/4)^{r(1-\exp(-ck))}$.
\end{thm}

%\begin{proof}
%We can write $S$ as $\sum_{v\in R} (1-X_v)$ where $X_v$ is the indicator
%variable that denotes that $X_v$ is matched. Note that the variables
%$1-X_v$ for each $v\in R$ are non-positively correlated. In
%particular, if $N(v)$ and $N(v')$ are disjoint, then $1-X_v$ and
%$1-X_{v'}$ are independent. Otherwise, $v$ not claiming any edges can
%only increase the probability that $v'$ has an edge from any vertex
%$u\in N(v)\cap N(v')$. Also note that the expected size of $S$ is
%$r(1-\exp(-ck))$ by Theorem \ref{original_result}. Therefore, we can
%apply Theorem \ref{negative_corr_chernoff} with $\delta=1$ to obtain
%the result.
%\end{proof}

For realistic scenarios where $r$ is very large, the above theorem gives very tight bounds on the size of the solution, also explaining the effectiveness of the simple sampling algorithm in such instances.

\subsection{The Greedy Algorithm}
\label{greedy}
We next analyze the natural greedy algorithm for constructing a $(c,a)$-recommendation
subgraph $H$ iteratively. In the following algorithm, we use $N(u)$ to refer to the neighbors
of a vertex $u$. \vspace{0.05in}

\begin{algorithm}[h]
  \SetAlgoLined
  \KwData{A bipartite graph $G=(L,R,E)$}
  \KwResult{A $(c,a)$-recommendation subgraph $H$}
  \For{$u$ in $L$}{
    $d[u] \leftarrow 0$
  }
  \For{$v$ in $R$}{
    $F \leftarrow \{u \in N(v) | d[u] < c\}$\;
    \If{$|F| \geq a$}{
      restrict $F$ to $a$ elements\;
      \For{$u$ in $F$}{
        $H \leftarrow H \cup \{(u,v)\}$\;
        $d[u] \leftarrow d[u]+1$\;
      }
   	}
  }
  \Return $H$\;
  \caption{The greedy Algorithm}
\end{algorithm}

The algorithm loops through each vertex in $R$, and considers each edge once.
Therefore, the runtime is $\Theta(|E|)$. Furthermore, the only data structure
we use is an array which keeps track of $\deg_H(u)$ for each $u\in L$, so the
memory consumption is $\Theta(|L|)$. Finally, we prove the following tight
approximation property of this algorithm.

\begin{thm}
The greedy algorithm gives a $1/(a+1)$-approximation to the $(c,a)$-graph
recommendation problem.
\end{thm}
\begin{proof}
Let $R_{GREEDY}, R_{OPT}\subseteq R$ be the set of vertices that have
degree $\geq a$ in the greedy and optimal solutions respectively. Note
that any $v \in R_{OPT}$ along with neighbors $\{u_1,\ldots u_a\}$
forms a set of candidate edges that can be used by the greedy
algorithm.
%So we can consider $R_{OPT}$ as a candidate pool for $R_{GREEDY}$.
Each selection of the greedy algorithm might result in
some candidates becoming infeasible, but it can continue as long as the candidate pool is not depleted.
Each time the greedy algorithm selects some vertex $v\in
R$ with edges to $\{u_1,\ldots, u_a\}$, we remove $v$ from the candidate pool.
Furthermore each $u_i$ could have degree $c$ in the optimal solution and used each of its edges to make a neighbor attain degree $a$. The greedy choice of an edge to $u_i$ requires us to remove such an edge to an arbitrary vertex $v_i\in R$ adjacent to $u_i$ in the optimal
solution, and thus remove $v_i$ from further consideration in the candidate pool.
%In other words, by using an edge of $u_i$, we force it to
%not use an edge it used to some other $v_i$, which might cause the
%degree of $v_i$ to go below $a$.
Therefore, at each step of
the greedy algorithm, we may remove at most $a+1$ vertices from
the candidate pool as illustrated in Figure 4. Since our candidate pool has size $OPT$, the
greedy algorithm can not stop before it has added $OPT/(a+1)$
vertices to the solution.
\end{proof}

\begin{figure}[H]
\label{fig:greedy}
\centering
\includegraphics[width=.39\textwidth]{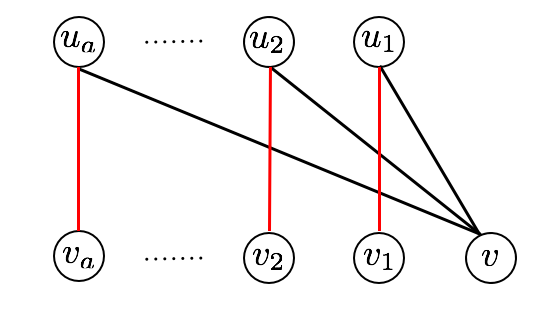}
\caption{One step of the greedy algorithm. When $v$ selects edges to $u_1,\ldots, u_a$, it can remove $v_1,\ldots, v_a$ from the pool of candidates that are available. The potentially invalidated edges are shown in red.}
\end{figure}

This approximation guarantee is as good as we can expect,
since for $a=1$ we recover the familiar $1/2$-approximation
of the greedy algorithm for matchings. Furthermore, even in the case of matchings ($a=1$),
randomizing the order in which the vertices are processed is still
known to leave a constant factor gap in the quality of the solution
~\cite{KarpVaziraniVazirani1990}. Despite this result, the greedy
algorithm fares much better when we analyze its expected performance.
Switching to the
{\bf Erd\"{o}s-Renyi model}~\cite{ErdosRenyi59} instead of the fixed degree model used in
the previous section, we now prove the near optimality of the
greedy algorithm for the $(c, a)$-recommendation subgraph problem.
Recall that in this model (sometimes referred to as $G_{n,p}$), each possible edge is inserted with probability $p$ independent of other edges. In our version $G_{l,r,p}$, we only add edges from $L$ to $R$ each with probability $p$ independent of other edges in this complete bipartite candidate graph.  For technical reasons, we need to assume that $lp \geq 1$ in the following theorem. However, this is a very weak assumption since $lp$ is simply the expected degree of a vertex $v\in L$. Typical values for $p$ for our applications will be $\Omega(\log(l)/l)$ making the expected degree $lp=\Omega(\log l)$.

\begin{thm}\label{greed-is-good}
Let $G=(L,R,E)$ be a graph drawn from the $G_{l,r,p}$ where $lp \geq 1$. If $S$ is the size of the $(c,a)$-recommendation subgraph produced by the greedy algorithm, then:
\[ \emph{\E}[S] \geq r - \frac{a(lp)^{a-1}}{(1-p)^a} \sum_{i=0}^{r-1}(1-p)^{l-\frac{ia}{c}}\]
\end{thm}
\begin{proof}
Note that if edges are generated uniformly, we can consider the
graph as being revealed to us one vertex at a time as the greedy
algorithm runs. In particular, consider the event $X_{i+1}$ that the
greedy algorithm matches the $(i+1)^{st}$ vertex it inspects. While,
$X_{i+1}$ is dependent on $X_1,\ldots, X_i$, the worst condition for
$X_{i+1}$ is when all the previous $i$ vertices were from the same
vertices in $L$, which are now not available for matching the
$(i+1)^{st}$ vertex. The maximum number of such invalidated vertices
is at most $\lceil ia/c \rceil$. Therefore, the bad event
is that we have fewer than $a$ of the at least $l-\lceil ia/c \rceil$ 
available vertices having an edge to this vertex. The probability of this
bad event is at most $\Pr[Y\sim Bin(l-\frac{ia}{c},p): Y < a]$, the
probability that a Binomial random variable with $l - \frac{ia}{c}$
trials of probability $p$ of success for each trial has less than $a$
successes. We can bound this probability by using a union bound and
upper-bounding $\Pr[Y\sim Bin(l-\frac{ia}{c},p): Y = t]$ for each 
$0 \leq t \leq a-1$. By using the trivial estimate that
$\binom{n}{i} \leq n^i$ for all $n$ and $i$, we obtain:

\begin{align*}
      \Pr[Y\sim Bin(l-\frac{ia}{c},p): Y = t]
&=    \binom{l-\frac{ia}{c}}{t} (1-p)^{l-\frac{ia}{c}-t}p^{t} \\
&\leq \left(l-\frac{ia}{c}\right)^t (1-p)^{l-\frac{ia}{c}-t} p^{t} \\
&\leq (lp)^t (1-p)^{l-\frac{ia}{c}-t} 
\end{align*}

Notice that the largest exponent $lp$ can take within the bounds of
our sum is $a-1$. Similarly, the smallest exponent $(1-p)$ can take
within the bounds of our sum is $l-\frac{ia}{c}-a+1$. Now applying
the union bound gives:

\begin{align*}
&           \Pr[Y\sim Bin(l-\frac{ia}{c},p): Y < a] \\
&\leq  \sum_{t=0}^{a-1} \Pr[Y\sim Bin(l-\frac{ia}{c},p): Y = t] \\
&\leq  \sum_{t=0}^{a-1} (lp)^t (1-p)^{l-\frac{ia}{c}-t} \\
&=     a(lp)^{a-1} (1-p)^{l-\frac{ia}{c}-a+1}
\end{align*}

Finally, summing over all the $X_i$ using the linearity of
expectation and this upper bound, we obtain

\begin{align*}
      \E[S]
&\geq r - \sum_{i=0}^{r-1} \E[\lnot X_i] \\
&\geq r - \sum_{i=0}^{r-1} \Pr[Y \sim Bin(l-\frac{ia}{c},p): Y < a] \\
&\geq r - a(lp)^{a-1}\sum_{i=0}^{r-1}(1-p)^{l-\frac{ia}{c}-a+1}
\end{align*}
\end{proof}

Asymptotically, this result explains why the greedy
algorithm does much better in expectation than $1/(a+1)$ guarantee we
can prove in the worst case. In particular, for a reasonable setting of
the right parameters, we can prove that the error term of our greedy
approximation will be sublinear.

\begin{thm}
Let $G=(L,R,E)$ be a graph drawn from the $G_{l,r,p}$ where $p = \frac{\gamma \log l}{l}$ for some $\gamma \geq 1$. Suppose that $c, a$ and $\epsilon>0$ are such that $lc=(1+\epsilon)ra$ and that $l$ and $r$ go to infinity while satisfying this relation. If $S$ is the size of the $(c,a)$-recommendation subgraph produced by the greedy algorithm, then

\[ \E[S] \geq r - o(r) \]
\end{thm}
\begin{proof}
We will prove this claim by applying Theorem \ref{greed-is-good}. Note that it suffices to prove that $(lp)^{a-1}\sum_{i=0}^{r-1}(1-p)^{l-\frac{ia}{c}} = o(r)$ since the other terms are just constants. We first bound the elements of this summation. Using the facts that $p = \frac{\gamma \log l}{l}$, $lc/a = (1+\epsilon)r$ and that $i<r$ throughout the summation, we get the following bound on each term:

\begin{align*}
       (1-p)^{l-\frac{ia}{c}} 
&\leq  \left(1-\frac{\gamma \log l}{l}\right)^{l-\frac{ia}{c}} \\
&\leq  \exp\left(-\frac{\gamma \log l}{l}\left(l-\frac{ia}{c}\right)\right) \\ 
&=     \exp\left((-\log l)\left(\gamma - \frac{ia}{lc}\right)\right) \\
&=     l^{-\gamma+\frac{ia}{lc}} = l^{-\gamma+\frac{i}{(1+\epsilon)r}} \\
&\leq  l^{-1+\frac{1}{1+\epsilon}} = l^{-\frac{\epsilon}{1+\epsilon}}
\end{align*}

Finally, we can evaluate the whole sum:

\begin{align*}
       (lp)^{a-1}  \sum_{i=0}^{r-1}(1-p)^{l-\frac{ia}{c}}
&\leq  \left(\log^{a-1}l\right) \sum_{i=0}^{r-1}l^{-\frac{\epsilon}{1+\epsilon}}\\
&\leq  \left(\log^{a-1}l\right) rl^{-\frac{\epsilon}{1+\epsilon}}\\
&=     \left(\log^{a-1}l\right) \frac{c}{(1+\epsilon)a}l^{1-\frac{\epsilon}{1+\epsilon}} = o(l)\\
\end{align*}

However, since $r$ is a constant times $l$, any function that is $o(l)$ is also $o(r)$ and this proves the claim.

\end{proof}

\subsection{The Partition Algorithm}
To motivate the partition algorithm, we first define optimal solutions for the recommendation subgraph problem.
\vs

{\bf Perfect Recommendation Subgraphs:} We define a \emph{perfect} $(c,a)$-recommendation subgraph on $G$ to be a subgraph $H$ such that
$deg_H(u)\leq c$ for all $u\in L$ and $deg_H(v)=a$ for
$\min(r,\lfloor cl/a \rfloor)$ of the vertices in $R$.
\vs

The reason we define perfect $(c,a)$-recommendation subgraphs is that when one
exists, it's possible to recover it in polynomial time using a min-cost
$b$-matching algorithm (matchings with a specified degree $b$ on each vertex)
for any setting of $a$ and $c$. However, implementations of $b$-matching
algorithms often incur significant overheads even over regular bipartite matchings.
This motivates a solution that uses regular bipartite matching algorithms to find
an approximately optimal solution given that a perfect one exists. \vs

We do this by proving a sufficient condition for perfect $(c,a)$-recommendation
subgraphs to exist with high probability in a bipartite graph $G$ under the
{\bf Erd\"os-Renyi model}~\cite{ErdosRenyi59} where edges are sampled uniformly and
independently with probability $p$. This argument then guides our formulation of
a heuristic that overlays matchings carefully to obtain $(c,a)$-recommendation
subgraphs. \vs

\begin{thm}\cite{Janson2011}
\label{random_matching_threshold}
Let $G$ be a bipartite graph drawn from $G_{n, n, p}$. If $p \geq \frac{\log n -
\log\log n}{n}$, then as $n\to\infty$,  the probability that G has a perfect
    matching approaches 1.
\end{thm}

We will prove that a perfect $(c,a)$-recommendation subgraph exists in
random graphs with high probability by building it up from $a$
matchings each of which must exist with high probability if $p$ is
sufficiently high. To find these matchings, we identify subsets of size
$l$ in $R$ that we can perfectly match to $L$. These subsets overlap,
and we choose them so that each vertex in $R$ is in $a$ subsets. 
%While the theorem is stated for
%the case when $a \leq c$, it applies equally well to the $a > c$
%case by partitioning $L$ instead of $R$ in the following proof.

\begin{thm}\label{perfect}
Let $G$ be a random graph drawn from $G_{l, r, p}$ with $p\geq a\frac{\log l-\log\log
l}{l}$ then the probability that $G$ has a perfect $(c, a)$-recommendation
subgraph tends to 1 as $l,r\to\infty$.
\end{thm}

\begin{proof}
We start by either padding or restricting $R$ to a set of $\frac{lc}{a}$ before we
start our analysis. If $r\geq\frac{lc}{a}$, then we restrict $R$
to an arbitrary subset $R'$ of size $\frac{lc}{a}$. Since induced subgraphs of
Erd\"{o}s-Renyi graphs are also Erd\"{o}s-Renyi graphs, we can instead
apply our analysis to the induced subgraph. Since the optimal
solution has size bounded above by $\frac{lc}{a}$ a perfect $(c,a)$-recommendation
subgraph in $G[L,R']$ will imply a perfect recommendation subgraph in $G[L,R]$. \vs

On the other hand, if $r \leq\frac{lc}{a}$, then we can pad $R$ with $\frac{lc}{a}-r$ 
dummy vertices and adding an edge from each such vertex to each vertex in $L$
with probability $p$. We call the resulting right side of the graph $R'$.
Note that $G[L,R']$ is still generated by the Erd\"{o}s-Renyi process. Further,
since the original graph $G[L,R]$ is a subgraph of this new graph, if we prove
the existence of a perfect $(c,a)$-recommendation subgraph in this new graph, it
will imply the existence of a perfect recommendation subgraph in $G[L,R]$. \vs

Having picked an $R'$ satisfying $|R'|=\frac{lc}{a}$, we pick an enumeration 
of the vertices in $R'=\{v_0,\ldots, v_{lc/a-1}\}$
and add each of these vertices into $a$ subsets as follows. Define
$R_i = \{v_{(i-1)l/a}, \ldots, v_{(i-1)l/a+l-1}\}$ for each $1\leq i\leq c$ where
the arithmetic in the indices is done modulo $lc/a$. Note both $L$ and all of
the $R_i$'s have size $l$. \vs

Using these new sets we define the graphs $G_i$ on the bipartitions
$(L, R_i)$. Since the sets $R_i$ are intersecting, we cannot define the
graphs $G_i$ to be induced subgraphs. However, note that each vertex $v\in R'$
falls into exactly $a$ of these subsets. \vs

Therefore, we can uniformly randomly assign each edge in $G$ to one of $a$ graphs among $\{G_1,\ldots, G_c\}$ it can fall into,
and make each of those graphs a random graph. In fact, while the different
$G_i$ are coupled, taken in isolation we can consider any single $G_i$ to be
drawn from the distribution $G_{l,l,p/a}$ since $G$ was drawn from $G_{l,r,p}$.
Since $p/a \geq (\log l - \log\log l)/l$ by assumption, we conclude by
Theorem~\ref{random_matching_threshold}, the probability that a particular
$G_i$ has no perfect matching is $o(1)$. \vs

If we fix $c$, we can conclude by a union bound that except
for a $o(1)$ probability, each one of the $G_i$'s has a perfect matching. By
superimposing all of these perfect matchings, we can see that every vertex in
$R'$ has degree $a$. Since each vertex in $L$ is in exactly $c$ matchings, each
vertex in $L$ has degree $c$. It follows that except for a $o(1)$ probability
there exists a $(c,a)$-recommendation subgraph in $G$.
\end{proof}

{\bf Approximation Algorithm Using Perfect Matchings:}
The above result now enables us to design a near linear time
algorithm with a $(1-\epsilon)$ approximation guarantee
to the $(c,a)$-recommendation subgraph problem by leveraging
combinatorial properties of matchings. In particular, we use
the fact a matching that does not have augmenting paths of
length $>2\alpha$ is a $1-1/\alpha$ approximation to the maximum
matching problem. We call this method the Partition Algorithm,
and we outline it below.

\begin{algorithm}[h]\label{partition_alg}
  \SetAlgoLined
  \KwData{A bipartite graph $G=(L,R,E)$}
  \KwResult{A (c,a)-recommendation subgraph $H$}
  $R' \leftarrow$ a random sample of $|L|c/a$ vertices from $R$\;
%Arda: In light of above comments. I am not sure you need to define R'
  Choose $G[L,R_1],\ldots,G[L,R_c]$ as in Theorem \ref{perfect}\;
  \For{$i$ in [1..n]} {
    $M_i \leftarrow$ A matching of $G[L,R_i]$ with no augmenting path of length $2c/\epsilon$\;
  }
  $H \leftarrow M_1\bigcup\ldots \bigcup M_c$\;
  \Return $H$\;
  \caption{The partition algorithm}
\end{algorithm}
\vspace{-.2cm}

\begin{thm}
Let $G$ be drawn from $G_{l,r,p}$ where $p \geq a\frac{\log l - \log\log l}{l}$.
Then Algorithm 3 finds a $(1-\epsilon)$-approximation
in $O(\frac{|E|}{\epsilon})$ time with probability $1-o(1)$.
\end{thm}
\begin{proof}
Using the previous theorem, we know that each of the graphs $G_i$ has a
perfect matching with high probability. These perfect matchings
can be approximated to a $1-\epsilon/c$ factor by finding matchings
that do not have augmenting paths of length $\geq 2c/\epsilon$
~\cite{LovaszPlummer1986}. This can be done for each $G_i$ in
$O(|E|c/\epsilon)$ time. Furthermore, the union of unmatched vertices
makes up an at most $c(\epsilon/c)$ fraction of $R'$, which proves the claim.
\end{proof}

Notice that if we were to run the augmenting paths algorithm to completeness
for each matching $M_i$, then this algorithm would take $O(|E||L|)$ time. We
could reduce this further to $O(|E|\sqrt{L})$ by using Hopcroft-Karp.
\cite{HopcroftKarp} \vs

Assuming
a sparse graph where $|E|=\Theta(|L|\log|L|)$, the time complexity of this algorithm is $\Theta(|L|^{3/2}\log|L|)$.  The space complexity
is only $\Theta(|E|) = \Theta(|L|\log|L|)$, but a large constant is hidden by
the big-Oh notation that makes this algorithm impractical in real test cases. 
\section{Experimental Results}

\subsection{Simulated Data}
We simulated performance of our algorithms on random graphs generated
by the graph models we outlined. 
%Arda: We mention fixed degree and Erdos-Renyi above - which one are you talking about here in the experiments?
In the following figures, each data
point is obtained by averaging the measurements over 100 random
graphs. We first present the time and space usage of these algorithms when
solving a $(10,3)$-recommendation subgraph problem in different sized graphs.
In all our charts, error bars are present, but too small to be noticeable.
Note that varying the value of $a$ and $c$ would only change space and time
usage by a constant, so these two graphs are indicative of time and space
usage over all ranges of parameters. The code used conduct these experiments 
can be found at \url{https://github.com/srinathsridhar/graph-matching-source} \vs

\begin{figure}
\centering
\begin{minipage}[h]{.48\textwidth}
\centering
\includegraphics[width=.99\textwidth]{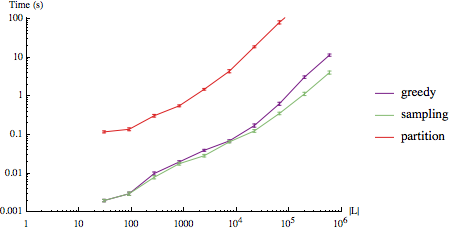}
\vspace{-0.2in}
\caption{Time needed to solve a (10,3)-recommendation problem in random graphs where $|R|/|L|=4$ (Notice the log-log scale.)}\label{fig:time_graph}
\end{minipage}
\vspace{.2cm}
\hspace{0cm}
\begin{minipage}[h]{.48\textwidth}
\centering
\includegraphics[width=.99\textwidth]{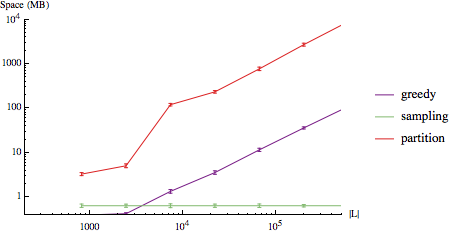}
\vspace{-0.2in}
\caption{Space needed to solve a (10,3)-recommendation problem in random graphs where $|R|/|L|=4$ (Notice the log-log scale.)}\label{fig:space_graph}
\end{minipage}
\vspace{-0.2in}
\end{figure} 

Recall that the partition algorithm split the graph into multiple graphs
and found matchings (using an implementation of Hopcroft-Karp ~\cite{HopcroftKarp}) 
in these smaller graphs which were then combined into
a recommendation subgraph. For this reason, a run of the partition
algorithm takes much longer to solve a problem instance than either the
sampling or greedy algorithms. It also takes significantly more memory as can be seen in Figures 5 and 6.

Compare this to greedy and sampling which both require a single pass over
the graph, and no advanced data structures. In fact, if the edges
of $G$ is pre-sorted by the edge's endpoint in $L$, then the sampling algorithm can be
implemented as an online algorithm with constant space and in constant time 
per link selection. Similarly, if the edges of $G$
is pre-sorted by the edge's endpoint in $R$, then the greedy algorithm can
be implemented so that the entire graph does not have to be kept in memory. In this
event, greedy uses only $O(|L|)$ memory. \vs

Next, we analyze the relative qualities of the solutions each method
produces.  Figures~\ref{fig:a=1:1} and \ref{fig:a=1:2} plot the
average performance ratio of the three methods compared to the trivial
upper bounds as the value of $c$, the number of recommendations
allowed is varied, while keeping $a = 1$.

They collectively show that the lower bound we calculated for the
expected performance of the sampling algorithm accurately captures its
behavior when $a=1$. Indeed, the inequality
we used is an accurate approximation of the expectation, up to lower
order terms, as is demonstrated in these simulated runs.  The random
sampling algorithm does well, both when $c$ is low and high, but
falters when $ck=1$. The greedy algorithm outperforms the
sampling algorithm in all cases, but its advantage vanishes as
$c$ gets larger. Note that the dip in the graphs when $cl=ar$, at
$c=4$ in Figure~\ref{fig:a=1:1} and $c=2$ in Figure~\ref{fig:a=1:2} is
expected and was previously demonstrated in
Figure~\ref{fig:simple_approx}.  The partition algorithm is immune to
this drop that affects both the greedy and the sampling algorithms,
but comes with the cost of higher time and space utilization.\vs

\begin{figure}
\centering
\begin{minipage}[h]{0.48\textwidth}
\centering
\includegraphics[width=0.99\textwidth]{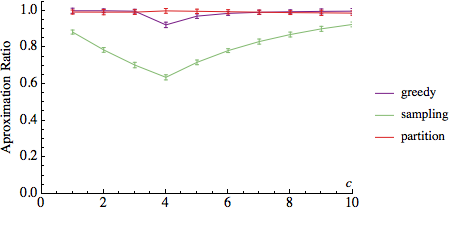}
\vspace{-1cm}
\caption{Solution quality for the $(c, 1)$-recommendation subgraph problem in graphs with $|L|=25$k, $|R|=100$k, $d=20$}\label{fig:a=1:1}
\end{minipage}
\vspace{.2cm}
\hspace{0cm}
\begin{minipage}[h]{0.48\textwidth}
\centering
\includegraphics[width=0.99\textwidth]{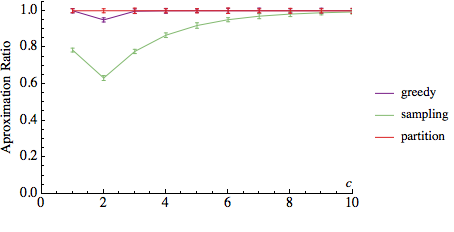}
\vspace{-1cm}
\caption{Solution quality for the $(c, 1)$-recommendation subgraph problem in graphs with $|L|=50$k, $|R|=100$k, $d=20$}\label{fig:a=1:2}
\end{minipage}
\vspace{-.8cm}
\end{figure}

In contrast to the case when $a=1$, the sampling algorithm
performs worse when $a>1$ but performs increasingly better with $c$ as
demonstrated by Figures~\ref{fig:a=2} and \ref{fig:a=4}. The greedy
algorithm continues to produce solutions that are nearly optimal,
regardless of the settings of $c$ and $a$, even beating the
partition algorithm with increasing values of $a$. Our simulations
suggest that in most cases, one can simply use our sampling method for solving the $(c, a)$-recommendation subgraph problem. In cases where the sampling is not suitable as flagged
by our analysis, we still find that the greedy performs adequately and
is also simple to implement. These two algorithms thus confirm to our
requirements we initially laid out for deployment in large-scale real
systems in practice.

\begin{figure}
\centering
\begin{minipage}[h]{0.48\textwidth}
\centering
\includegraphics[width=0.99\textwidth]{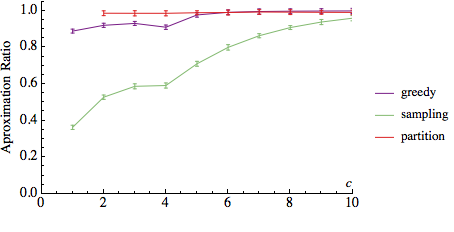}
\vspace{-1cm}
\caption{Solution quality for the $(c, 2)$-recommendation subgraph problem in graphs with $|L|=50$k, $|R|=100$k, $d=20$}\label{fig:a=2}
\end{minipage}
\vspace{.2cm}
\hspace{0cm}
\begin{minipage}[h]{0.48\textwidth}
\centering
\includegraphics[width=0.99\textwidth]{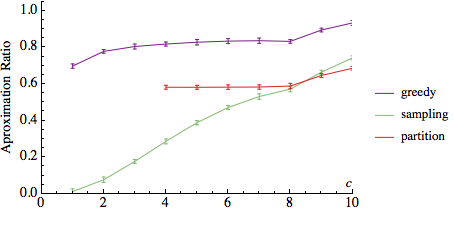}
\vspace{-1cm}
\caption{Solution quality for the $(c, 4)$-recommendation subgraph problem in graphs with $|L|=50$k, $|R|=100$k, $d=20$}\label{fig:a=4}
\end{minipage}
\vspace{-.8cm}
\end{figure}

To summarize, our synthetic experiments show the following strengths of each algorithm: \vs

\textbf{Sampling Algorithm:} Sampling uses little to no memory and can
be implemented as an online algorithm. If keeping the underlying graph in
memory is an issue, then chances are this algorithm will do well while only needing
a fraction of the resources the other two algorithms would need. \vs

\textbf{Partition Algorithm:} This algorithm does well, but only when $a$ is small.
In particular, when $a=1$ or 2, partition seems to be the best algorithm, but the quality
of the solutions degrade quickly after that point. However this performance comes at
expense of significant runtime and space. Since greedy performs almost as well without
requiring large amounts of space or time, partition is best suited for instances where
$a$ is low the quality of the solution is more important than anything else. \vs

\textbf{Greedy Algorithm:} This algorithm is the all-round best performing algorithm we tested.
It only requires a single pass over the data thus very quickly,  and uses
relatively little amounts of space enabling it run completely in memory for graphs with
as many as tens of millions of edges. It is not as fast as sampling or accurate as partition
when $a$ is small, but it has very good performance over all parameter ranges. 

\subsection{Real Data}

We now present the results of running our algorithms on several real
datasets. In the graphs that we use, each node corresponds to a single
product in the catalog of a merchant and the edges connect similar
products. For each product up to 50 most similar products were
selected by a proprietary algorithm of BloomReach that uses text-based
features such as keywords, color, brand, gender (where applicable) as
well as user browsing patterns to determine the similarity between
pairs of products. Such algorithms are commonly used in e-commerce websites
such as Amazon, Overstock, eBay etc to display the most related products 
to the user when they are browsing a specific product. \vs

Two of the client merchants of BloomReach presented here had
moderate-sized relation graphs with about $10^5$ vertices and $10^6$
input edges (candidate recommendations); the remaining merchants (3, 4
and 5) have on the order of $10^6$ vertices and $10^7$ input edges
between them.  We estimated an upper bound on the optimum solution by
taking the minimum of $|L|c/a$ and the number of vertices in $R$ of
degree at least $a$. Figures~\ref{fig:real_a=1},~\ref{fig:real_a=2}
and~\ref{fig:real_a=3} plot the average of the optimality percentage
of the sampling, greedy and partition algorithms across all the
merchants respectively. Note that we could only run the partition
algorithm for the first two merchants due to memory constraints. \vs

\begin{figure}[h]
\centering
\begin{minipage}[h]{0.48\textwidth}
\centering
\includegraphics[width=0.99\textwidth]{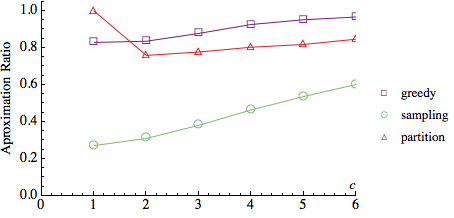}

\caption{Solution quality for the $(c, 1)$-recommendation subgraph problem in retailer data}
\label{fig:real_a=1}
\end{minipage}

\hspace{0cm}
\begin{minipage}[h]{0.48\textwidth}
\centering
\includegraphics[width=0.99\textwidth]{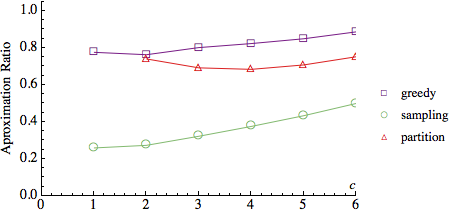}

\caption{Solution quality for the $(c, 2)$-recommendation subgraph problem in retailer data}
\label{fig:real_a=2}
\end{minipage}

\hspace{0cm}
\begin{minipage}[h]{0.48\textwidth}
\centering
\includegraphics[width=0.99\textwidth]{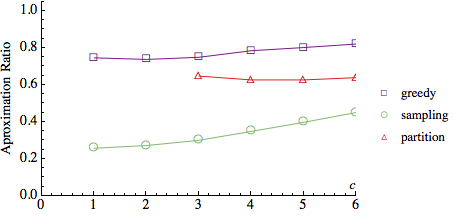}

\caption{Solution quality for the $(c, 3)$-recommendation subgraph problem in retailer data}
\label{fig:real_a=3}
\end{minipage}
\vspace{-.3cm}
\end{figure} \vs

From these results, we can see that that greedy performs exceptionally
well when $c$ gets even moderately large.  For the realistic value of
$c=6$, the greedy algorithm produced a solution that was 85\% optimal
for all the merchants we tested. For several of the merchants, its
results were almost optimal starting from $a=2$. \vs

The partition method is also promising, especially when the $a$ value
that is targeted is low. Indeed, when $a=1$ or $a=2$, its performance
is comparable or better than greedy, though the difference is not as pronounced as
it is in the simulations. However, for larger values of $a$ the partition
algorithm performs worse. 

The sampling algorithm performs mostly well on real data,
especially when $c$ is large. It is typically worse than
greedy, but unlike the partition algorithm, its performance improves
dramatically as $c$ becomes larger, and its performance does not worsen
as quickly when $a$ gets larger. Therefore, for large $c$ 
sampling becomes a viable alternative to greedy mainly in cases where the
linear memory cost of the greedy algorithm is too prohibitive.

\vspace{-0.1in}
\section{Summary and Future Work}
We have presented a new class of structural recommendation problems
cast as subgraph selection problems, and analyzed three algorithmic
strategies to solve these problems because graph matching algorithms can be
prohibitive to implement in real-world scenarios. The sampling method is most
efficient, the greedy approach trades off computational cost with
quality, and the partition method is effective for smaller problem
sizes. We have proved effective theoretical bounds on the quality
of these methods, and also substantiated them with experimental
validation both from simulated data and real data from
retail web sites. Our findings have been very useful in the
deployment of effective structural recommendations in web relevance
engines that drive many of the leading websites of popular retailers. \vs

Our sampling method and its analysis extends to more general
models of random graphs: in one version, we can consider
hierarchical models that take into account the product hierarchy
trees under which the pages in $L$ and $R$ are situated. A second
version considers a Cartesian product model where the pages in $L$
and $R$ are partitioned into closely related blocks and the graph
induced between every pair of left-right blocks follows a fixed
degree random model. A third variant models the potential flow of
customer traffic over each possible recommended edge from a left to
right page with nonnegative weights, and the resulting problem is
to find a subgraph where the number of right nodes with at least a
certain minimum amount of recommended traffic. Validating these
more general models by fitting real life data to them as well as
corroborating the performance of various methods in simulated and
real data for these models could yield an even better understanding
of our suggested algorithmic strategies for the 
recommendation subgraph problem.\vs

{\bf Acknowledgments:} We thank Alan Frieze and Ashutosh Garg for helpful
discussions.

\vspace{-.05in}

\setcounter{page}{1}
\bibliography{main.bbl}{}
\bibliographystyle{plain}

\end{document}